\documentclass[conference]{ieeeconf}
\pdfoutput=1

\IEEEoverridecommandlockouts
\usepackage{cite}
\usepackage{amsmath,amssymb,amsfonts}
\usepackage{algorithmic}
\usepackage{graphicx}
\usepackage{textcomp}
\usepackage{xcolor}
\usepackage{comment}
\usepackage{graphicx}
\usepackage{epsfig}
\usepackage{color}
\usepackage{enumerate}
\usepackage{jphmacros2e}
\usepackage{balance}

\newcommand{\bm}[1]{\mbox{$\boldmath{#1}$}}
\newcommand{\mc}[1]{\mathcal{#1}}
\newcommand{\field}[1]{{\mathbb{#1}}}
\newcommand{\bea}{\begin{eqnarray}}
\newcommand{\eea}{\end{eqnarray}}
\newcommand{\beas}{\begin{eqnarray*}}
\newcommand{\eeas}{\end{eqnarray*}}
\newcommand{\leftm}{\left[\begin{array}}
\newcommand{\rightm}{\end{array}\right]}
\newcommand{\reals}{\mbox{$\mathbb R$}}
\newcommand{\complex}{\mbox{$\mathbb C$}}
\newcommand{\ones}{\textbf{1}}
\newcommand{\zeros}{\textbf{0}}
\newcommand{\mA}{\mathcal{A}}
\newcommand{\mI}{\mathcal{I}}
\newcommand{\mC}{\mathcal{C}}
\newcommand{\mL}{\mathcal{L}}
\newcommand{\mV}{\mathcal{V}}
\newcommand{\mN}{\mathcal{N}}
\newcommand{\blue}[1]{{\color{blue}#1}}
\newcommand{\red}[1]{{\color{red}#1}}
\newcommand{\xa}{x_{\text{aug}}}
\DeclareMathOperator*{\argmax}{arg\,max}
\DeclareMathOperator*{\argmin}{arg\,min}
\DeclareMathOperator*{\argsup}{arg\,sup}
\DeclareMathOperator*{\softmax}{softmax}
\newcommand{\T}{\textrm{T}}
\newcommand{\dt}{\textrm{d}}

\newtheorem{thm}{Theorem}

\newtheorem{Problem}[thm]{Problem}

\newtheorem{remark}[thm]{Remark}

\def\BibTeX{{\rm B\kern-.05em{\sc i\kern-.025em b}\kern-.08em
    T\kern-.1667em\lower.7ex\hbox{E}\kern-.125emX}}
\begin{document}

\title{Temporal-Logic-Based Intermittent, Optimal, and Safe Continuous-Time Learning for Trajectory Tracking
\author{Aris~Kanellopoulos$^1$,
        ~Filippos~Fotiadis$^1$,
         ~Chuangchuang~Sun$^2$,
             ~Zhe~Xu$^3$,\\
             ~Kyriakos~G.~Vamvoudakis$^1$,
        ~Ufuk~Topcu$^4$,
~Warren~E.~Dixon$^5$
        }
\thanks{$^1$A.~Kanellopoulos, F.~Fotiadis, and K.~G.~Vamvoudakis are with the Daniel Guggenheim School of Aerospace Engineering, Georgia Institute of Technology, Atlanta, Georgia, USA, $30332$, USA, e-mail: \{ariskan, ffotiadis, kyriakos\}@gatech.edu.}
\thanks{$^2$Chuangchuang~Sun is with the Department of Aeronautics and Astronautics, Massachusetts Institute of Technology, Cambridge, MA $02139$,
USA, email: ccsun1@mit.edu.}

\thanks{$^3$Zhe~Xu is with the School for Engineering of Matter, Transport, and Energy, Arizona State University, Tempe, AZ $85287$, email: xzhe1@asu.edu.}
\thanks{$^4$Ufuk~Topcu are with the Department of Aerospace Engineering and Engineering Mechanics, and the Oden Institute for Computational Engineering
and Sciences, University of Texas, Austin,
TX $78712$, USA, email: utopcu@utexas.edu.}
\thanks{$^5$Warren~E.~Dixon is with the the Department of Mechanical and Aerospace Engineering, University of Florida, Gainesville, FL $32611-6250$, USA, e-mail: wdixon@ufl.edu}

\thanks{This work was supported in part, by ARO under grant No. W$911$NF-$19-1-0270$, by ONR Minerva under grant No. N$00014-18-1-2160$, and by NSF under grant Nos. CAREER CPS-$1851588$ and S\&AS $1849198$.}}

\maketitle

\begin{abstract}
In this paper, we develop safe reinforcement-learning-based controllers for systems tasked with accomplishing complex missions that can be expressed as linear temporal logic specifications, similar to those required by search-and-rescue missions.
We decompose the original mission into a sequence of tracking sub-problems under safety constraints. We impose the safety conditions by utilizing barrier functions to map the constrained optimal tracking problem in the physical space to an unconstrained one in the transformed space.
Furthermore, we develop policies that intermittently update the control signal to solve the tracking sub-problems with reduced burden in the communication and computation resources. Subsequently, an actor-critic algorithm is utilized to solve the underlying Hamilton-Jacobi-Bellman equations. Finally, we support our proposed framework with stability proofs and showcase its efficacy via simulation results.
\end{abstract}


\section{Introduction}

The problem of assured autonomy is challenging when the dynamics are complex and time-varying, the models are
largely unknown, the goals are dynamic and complex, and the environment is unknown and possibly adversarial. 
For learning-based systems the question of how to design, implement, and maintain high-confidence,
high-performance, and dynamically-configured secure policies for avoiding unsafe operating regions is of
paramount importance. 
Recognizing that reinforcement learning (RL) \cite{sutton2018reinforcement,kiumarsi2017optimal,vrabie2013optimal,kamalapurkar2018reinforcement} is an important component of assured autonomy, we
focus on learning-enabled systems.

In the context of control systems, a large class of objectives can be modeled as a requirement to follow a desired trajectory, rather than regulate the state of the system. Furthermore, certain scenarios involve a number of different reference trajectories and endpoints through which the system needs to advance. One scenario indicative of this procedure is found in search and rescue missions in which an autonomous vehicle follows specific search trajectories, with possible intermediate stops for recharging, before eventually returning to a specific location after certain conditions have been met \cite{lin2009uav,li2018planning}.

Temporal logic specifications \cite{emerson1990temporal} offer a systematic way of describing the different modes of operation of such system as well as the ways that those modes are interconnected through time. Similarly, to facilitate the use of those methods in high-risk environments, energy expenditure should be minimized.

The complexity of the required tasks that autonomous systems will have to accomplish, especially when time critical missions are considered,
leads to the need for the designer to develop control strategies that alleviate the burden on the communication and computation resources of the system. This, in turn, led to the introduction of event-triggered control \cite{heemels2012introduction}. Event-triggering mechanisms create another layer of security and safety in the autonomous system, since they minimize the opportunities for erroneous signals to affect the system; whether they are caused by environmental disturbances or malicious attacks.

\subsection*{Related Work}
The problem of safe learning has been in the forefront in recent years. The authors in \cite{ames2016control} combine control barrier functions with Lyapunov control functions via quadratic programs and introduce new classes of barrier functions to construct safe controllers. 
The framework of control barrier functions is applied to the problem of autonomous navigation of traffic circles in \cite{konda2019provably}.
In \cite{kimmel2015active}, the authors deal with the problem of human-robot interaction, specifically confronting safety issues via a dynamic invariance control framework.
The authors of \cite{fisac2018general}, investigate Hamilton-Jacobi-based reachability methods, thus constructing a switching controller that guarantees the system will remain within a predefined region.
The authors of \cite{greene2020sparse} develop an approximate online adaptive solution to an optimal control problem under safety constraints with the use of barrier functions and sparse learning.
The use of temporal logic specifications for safety has been studied in \cite{sun2020continuous}, where regulation problems have been solved via RL techniques while temporal logic specifications are guaranteed to be satisfied.
In \cite{muniraj2018enforcing}, deep Q-learning was leveraged to guarantee given specifications in a Markov decision process framework.
The motion-planning problem in a multi-robot system under temporal logic specifications is investigated in \cite{saha2014automated} where the authors use a library of motion primitives to accomplish the given mission.
Finally, the authors of \cite{sadigh2016safe} consider safety guarantees for control applications via probabilistic signal temporal logic.

The framework of event-triggered control has been extensively investigated in the literature, e.g., in \cite{heemels2012periodic,girard2014dynamic,heemels2012introduction,cheng2017event}. While the authors in \cite{donkers2010output} expanded the framework to take into account output feedback, most implementations remained static. To further enhance the capabilities of autonomous systems, event-triggered control has been used in tandem with RL techniques in various scenarios. In \cite{vamvoudakis2014event}, we have brought together intermittent mechanisms to alleviate the burden of an actor-critic framework. This was later extended for systems with unknown dynamics under a Q-learning framework in \cite{vamvoudakis2016event}.
The authors of \cite{xu2019controller} developed a controller with intermittent communication for a multi-agent system, whose dwell-time conditions, needed for stability, as well as safety constraints were expressed by metric temporal logic specifications.
Finally, event-triggering mechanisms were employed to the problem of autonomous path planning \cite{kontoudis2020online}. Compared to the aforementioned works, and to the best of our knowledge, this is the first time that an optimal, safe, and intermittent learning framework combined with formal methods is used in a continuous-time framework for tracking a family of trajectories.

\paragraph*{Contributions}
The contributions of this work are threefold. Firstly, we formulate a system tasked with accomplishing a mission consisting of regulation and tracking sub-problems. We decouple the problems through the use of a finite state automaton (FSA), which also models safety constraints. Secondly, we apply a barrier function-based transformation on the system and the required trajectories, thus allowing us to map the original problem into a series of optimal tracking sub-problems. Finally, we employ an actor-critic framework to solve the underlying tracking problems in a model-free data-driven fashion.

\paragraph*{Structure}
The rest of the paper is structured as follows. In Section II, we formulate the problem, introduce the temporal logic methodology that facilitates the construction of the tracking sub-problems and we employ barrier functions to map the original problem to a transformed state space that guarantees safety. In Section III, the optimal tracking sub-problems are formulated in the transformed state space and event-triggered controllers are derived. In Section IV, we introduce an actor-critic framework that solves the sub-problems. In Section V, we present simulation results for a system tasked with accomplishing a mission consisting of sequential regulation and tracking and, finally, in Section VI, we conclude the paper and discuss potential future directions.

\section{Problem Statement}
\label{sec_problem}
Consider the time-invariant control-affine nonlinear system

\bea\label{eq:system}
&&\dot{x} = f(x) + g(x)u,\nonumber\\
&&x(0) = x_0,~ t\geq 0,
\eea
where $x \in\reals^n, \ u\in\reals^m$ are the states and the control input, respectively.

The system \eqref{eq:system} is desired to achieve certain goals, constrained by temporal logic specifications (which we will describe later), by following one of the trajectories given by the family of exosystems
\bea\label{eq:system-tracked}
\dot{z}_i = h_i(z_i), ~z_i(0) = z_{i,0}, ~\forall t \ge 0, ~i \in \mI,
\eea
where $z_i:\reals_+\rightarrow\reals^n$ is the $i$-th candidate of the desired trajectories to be tracked, $h_i:\reals^n \to \reals^n$ is a Lipschitz continuous function with $h_i(0) = 0$, and $\mI$ is the set of the trajectories to be tracked.

\subsection{Linear Temporal Logic Syntax and Semantics}
We consider \textit{syntactically co-safe linear temporal logic} (co-safe LTL) and \textit{syntactically safe linear temporal logic} (safe LTL) formulas \cite{Lahijanian2016, Xu2019} for the specifications. Let $\mathbb{B} = \{\textrm{True}, \textrm{False}\}$ be the Boolean domain. A time set $\mathbb{T}$ is $\mathbb{R}_{>0}$. A set $AP$ is a set of \textit{atomic predicates}, each of which is a mapping $\reals^n\times\mathbb{T}\rightarrow\mathbb{B}$.

The syntax of a co-safe LTL formulas can be recursively defined as follows
\begin{align}\nonumber
&\phi :=\top~ | ~p ~|~ \lnot p ~|~  \phi \wedge \phi ~|~ \phi \vee \phi~|~ \bigcirc \phi ~|~ \Diamond \phi ~|~ \phi \mathcal{U} \phi,
\end{align}
where $\top$ stands for the Boolean constant True; $p\in AP$ is an atomic
predicate; $\lnot$ (negation), $\wedge$ (conjunction), and $\vee$ (disjunction)
are standard Boolean connectives; $\bigcirc$ (next), $\Diamond$ (eventually), and $\mathcal{U}$ (until) are temporal operators. 

The syntax of a safe LTL formulas can be recursively defined as follows,
\begin{align}\nonumber
&\phi :=\top~ | ~p ~|~ \lnot p ~|~  \phi \wedge \phi ~|~ \phi \vee \phi~|~ \bigcirc \phi ~|~ \Box \phi,
\end{align}
where $\Box$ (always) is a temporal operator.

We refer the readers to Sec. II-B of \cite{Lahijanian2016} for the Boolean semantics of co-safe and safe LTL formulas. For a co-safe LTL formula $\phi$, one can construct a finite state automaton (FSA) that accepts precisely the proposition sequences (i.e., \textit{words}) that satisfy $\phi$. For a safe LTL formula $\phi$, one can construct an FSA that accepts precisely the proposition sequences (i.e., \textit{words}) that violate $\phi$.

For example, a co-safe LTL specification $\phi_{\textrm{c}}=\Diamond p_2\  \wedge \ (\lnot p_2 \mathcal{U} p_1)$ can express that ``an unmanned aerial vehicle (UAV) should track a certain trajectory $z_1$ (see \eqref{eq:system-tracked}) before tracking another trajectory $z_2$'', where $p_1=(\vert\vert x(t)-z_1(t)\vert\vert\le \epsilon)$, and $p_2=(\vert\vert x(t)-z_2(t)\vert\vert\le \epsilon)$, and $\epsilon\in\mathbb{R}_{>0}$ is a threshold for tracking error.



\begin{Problem}\label{Problem:safety}
For the system given in \eqref{eq:system} and $\lambda >0$, find a control policy such that the closed-loop system has a stable equilibrium point, the control input satisfies $\|u\| \le \lambda$, and the trajectory of state $x$ satisfies an LTL specification $\phi=\phi_{\textrm{c}}\wedge\phi_{\textrm{s}}$, where $\phi_{\textrm{c}}$ is a co-safe LTL formula and $\phi_{\textrm{s}}$ is a safe LTL formula. \frqed
\end{Problem}

In this paper, for simplicity we consider the safe LTL formula $\phi_{\textrm{s}}$ to be in the form of $\phi_{\textrm{s}}=\Box p$, where $p$ is an atomic predicate.

\subsection{Decomposition of LTL Specifications}\label{sec:dec}
Given that $\phi=\phi_{\textrm{c}}\wedge\phi_{\textrm{s}}$ and $\phi_{\textrm{s}}=\Box p$, we construct an FSA that accepts precisely the proposition sequences that satisfy $\phi_{\textrm{c}}$ and manually divide Problem \ref{Problem:safety} into a series of sub-problems based on the states of the constructed FSA. For example, consider the LTL specification $\phi=\phi_{\textrm{c}}\wedge\phi_{\textrm{s}}$, where $\phi_{\textrm{c}}=\Diamond p_2\  \wedge \ (\lnot p_2 \mathcal{U} p_1)$ and $\phi_{\textrm{s}}=\Box p_3$, where $p_1$, $p_2$ and $p_3$ are different atomic predicates. Based on $\phi_{\textrm{c}}$ one can construct an FSA as shown in Figure~\ref{f:automaton}. Let $\mathcal{O}(p)$ denote the set of time-dependent states that satisfy an atomic predicate $p$. We assume that $\mathcal{O}(p_1)\cap\mathcal{O}(p_2)=\emptyset$ if $p_1$ and $p_2$ are different atomic predicates. Then, the only path from the initial FSA state $q_0$ to the final state $q_f$ is $q_0 \to q_1 \to q_f$, where the accepting state $q_f$ indicates that $\phi_{\textrm{c}}$ is satisfied. 
There are two resulting two-point boundary-value problem (TPBVP) sub-problems, when the state of the FSA transitions from $q_0$ to $q_1$, and transitions from $q_1$ to $q_f$, respectively. Specifically, when the state of the FSA transitions from $q_0$ to $q_1$, the two boundary conditions for the first TPBVP sub-problem are the initial state $x_0$ and $\mathcal{O}(p_1)$, respectively. In this first TPBVP sub-problem, the safety constraint can be encoded as $x\in\mathcal{O}(\neg p_2)\cap\mathcal{O}(p_3)$, as the state can not reach $\mathcal{O}(p_2)$ until it reaches $\mathcal{O}(p_1)$ according to $\phi_{\textrm{c}}$, and the state must always be in $\mathcal{O}(p_3)$ according to $\phi_{\textrm{s}}$. Similarly, when the state of the FSA transitions from $q_1$ to $q_f$, the two boundary conditions for the second TPBVP sub-problem are $\mathcal{O}(p_1)$ and $\mathcal{O}(p_2)$, respectively. In this second TPBVP sub-problem, the safety constraint can be encoded as $x\in\mathcal{O}(p_3)$ according to $\phi_{\textrm{s}}$. After $q_f$ is reached in the FSA, the state only needs to stay in $\mathcal{O}(p_3)$ until the end of time (according to $\phi_{\textrm{s}}$).

Note that generally for a complex FSA, it may not be straightforward to select a unique path from an initial FSA state to a final FSA state. Hence, we can view the FSA as a directed graph, (approximately) estimate the edge weights (i.e., distance), and then select a path by solving a shortest path problem. So far, we have finished decomposing the LTL specification into a sequence of sub-problems, that as a whole will eventually satisfy the original LTL specification.


\begin{figure}[!ht]
	\vspace{-0.1cm}
	\hspace{0.5cm}
	\includegraphics[scale=0.6]{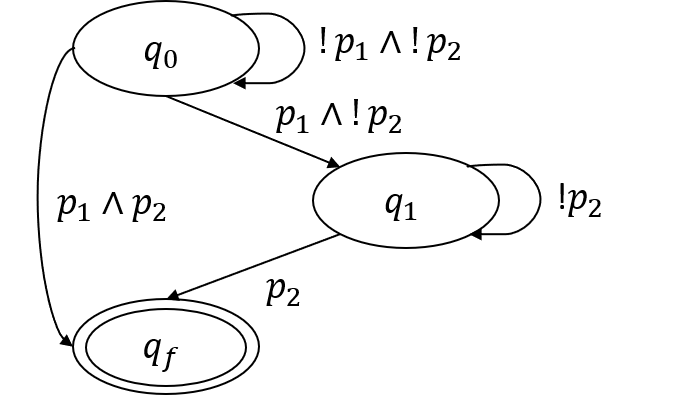}
	\vspace{-0cm}
	\caption{Finite state automaton generated by a co-safe LTL formula $\phi_{\textrm{c}}=\Diamond p_2\  \wedge \ (\lnot p_2 \mathcal{U} p_1)$.}
	\label{f:automaton}
\end{figure}

\subsection{Sub-Problem Tracking a Certain Trajectory}
In this paper, given that $\phi=\phi_{\textrm{c}}\wedge\phi_{\textrm{s}}$, we consider the predicates in the co-safe LTL formula $\phi_{\textrm{c}}$ to be in the form of $p=(\vert\vert x(t)-z_i(t)\vert\vert\le \epsilon), ~i\in \mI$. In this way, a certain trajectory $z_i, ~i\in \mI$ has to be tracked in each sub-problem decomposed from Problem \ref{Problem:safety}. Moreover, as described in Sec. \ref{sec_problem}-B, the system state $x$ is subject to some safety constraints in each sub-problem, as $x\in \mathcal{Q}$ needs to hold in each sub-problem, where $\mathcal{Q}= \{x\in\reals^n|c \le Ax + r \le C\}$, $A = [a_1~a_2~\ldots~a_m]^\T\in\reals^{m\times n}$,
with $a_i \in\reals^{n},~\forall i\in\lbrace 1,\dots,m\rbrace$,
$r = [r_1,\ldots,r_m]^{\textrm{T}}\in\reals^m$,
$c = [c_1,\ldots,c_m]^{\textrm{T}}\in\reals^m$ and $C = [C_1,\ldots,C_m]^{\textrm{T}} \in\reals^m$. The TPBVP sub-problem can thus be summarized as follows.

\begin{Problem}\label{Problem:safety2}
(sub-problem) For the system \eqref{eq:system}, find a control input $u$ that is constrained to satisfy $\|u\| \le \lambda$, so that the system state $x$ tracks a certain trajectory $z_i, i\in \mathcal{I}$ (i.e., $(\vert\vert x(t)-z_i(t)\vert\vert\le \epsilon)$ for a given threshold tracking error $\epsilon$) and the state $x$ remains in the set $\mathcal{Q}$, given an initial condition $x(0)=x_0$. \frqed
\end{Problem}

\begin{remark}
As a sub-problem of Problem~\ref{Problem:safety}, Problem~\ref{Problem:safety2} translates a temporal logic constraint $\phi$ into a more detailed and time-invariant constraint in the current time interval.
\end{remark}

To guarantee the safety specifications, we need to transform the system \eqref{eq:system}, which is constrained by the safety specifications $\mathcal{Q}$, into an unconstrained one. To this end, we design the following barrier function:
\bea\label{eq:barrier}
b(q,{c}_0,{C}_0) = \log\Big(\frac{C_0}{c_0}\frac{c_0-q}{C_0-q}\Big),\forall p\in(c_0,C_0),
\eea
where ${c}_0 < 0 < {C}_0$. The barrier $b(q,{c}_0, {C}_0)$ is invertible in the interval $({c}_0, {C}_0)$, and its inverse is given by
\bea\label{eq:barrier_inverse}
b^{-1}(y,{c}_0, {C}_0) = c_0C_0\frac{e^{\frac{y}{2}} - e^{-\frac{y}{2}}}{c_0e^{\frac{y}{2}} - C_0e^{-\frac{y}{2}}},\forall y\in \reals,
\eea
with dynamics
\bea\label{eq:barrier_inverse_derivatives}
\frac{\dt b^{-1}(y,{c}_0, {C}_0)}{\dt y} = \frac{C_0c_0^2 - c_0C_0^2}{c_0^2e^y - 2c_0C_0 + C_0^2e^{-y}}.
\eea
We now use \eqref{eq:barrier}-\eqref{eq:barrier_inverse_derivatives} to perform the transformation of \eqref{eq:system}. In particular, let us define
\bea\label{eq:transformation}
s_i &=& b(q_i(x),c_i,C_i) \nonumber\\
q_i(x) &=& b^{-1}(s_i,c_i,C_i) \\
q_i(x) &=& a_i^\textrm{T}x + r_i, \forall i=1,\ldots,m. \nonumber
\eea
Through the use of the chain rule, we obtain
\bea
\frac{\dt q_i(x)}{\dt t} = a_i^\textrm{T}\dot{x} = \frac{\dt b^{-1}(s_i,c_i,C_i)}{\dt s_i}\frac{\dt s_i}{\dt t},\nonumber
\eea
which yields
\bea\label{eq:transformation2}
\frac{\dt s_i}{\dt t} = \frac{1}{\frac{\dt b^{-1}(s_i,c_i,C_i)}{\dt s_i}}a_i^\textrm{T}\dot{x}.
\eea
Additionally, from \eqref{eq:transformation}, we have that
\bea\label{eq:transformation5}
Ax + r = b^{-1}(s,c,C),
\eea
where $b^{-1}(s,c,C) = [b^{-1}(s_1,c_1,C_1),\ldots,b^{-1}(s_m,c_m,C_m)]^\textrm{T} \in\reals^m$, hence we conclude that
\bea\label{eq:transformation6}
x  = (A^\T A)^{-1}A^\T(b^{-1}(s,c,C) - r).
\eea

\begin{remark}
To make $A^\T A\in\reals^{n}$ invertible, $A$ is assumed to be full column rank. This assumption is not restrictive, since we can make it hold by adding sufficiently large trivial bounds on the states, i.e., $-M_i \le x_i  \le M_i$, with $M_i\gg0$, $\forall i=1,\ldots,n$.
\frqed
\end{remark}

Combining \eqref{eq:transformation2} with the dynamics of $x$ yields the unconstrained subsystem dynamics, $\forall i=1,\ldots,n$,
\bea\label{eq:transformation7}
\frac{\dt s_i}{\dt t} &=& \frac{1}{\frac{\dt b^{-1}(s_i,c_i,C_i)}{\dt s_i}}a_i^\textrm{T}\dot{x} \\
&=&\frac{1}{\frac{\dt b^{-1}(s_i,c_i,C_i)}{\dt s_i}}a_i^\textrm{T} \big(f(x) + g(x)u \big),\ t\geq 0.\nonumber
\eea
 In \eqref{eq:transformation7}, the constrained state $x$ can be written with respect to the unconstrained state $s=[s_1,\ldots,s_n]^\textrm{T}$ as in \eqref{eq:transformation6}. Therefore, the subsystems \eqref{eq:transformation7} can be described in the compact form
\bea\label{eq:transformation8}
\dot{s} = F(s) + G(s)u,\ t\geq 0.
\eea
where $F:\reals^n\rightarrow\reals^n$, $G:\reals^n\rightarrow\reals^{n\times m}$.

To achieve optimal tracking while satisfying the required safety constraints, we shall solve Problem $2$ in the transformed $s$-domain both for the system---employing the dynamics given by \eqref{eq:transformation8}---as well as the image of the target trajectory in the $s$-domain.
Thus, let the transformed dynamics of $z\in z_\mathcal{I}$ be given by

\begin{align}
&\dot{z}_s(t)= f_d(z_s(t)),\ z_s(0)=z_0,\ t\geq 0, \label{eq:exo1}
\end{align}
where $z_s(t)\in\reals^n$ denotes the bounded desired trajectory in the $s$-domain, and $f_d$ is a Lipschitz continuous function, with $f_d(0)=0$, which yields the dynamics of $z_s$. The function $f_d$ can be derived by following the same procedure that was used to transform the state $x$ into $s$.

We may now define a tracking error $e_s(t)=s(t)-z_s(t)\in\mathbb{R}^n$ in the $s$-domain, $\forall t\ge 0$, with associated dynamics given by
\begin{equation*}
    \dot{e}_s(t)=F(e_s(t)+z_s(t))-f_d(z_s(t)).
\end{equation*}
Hence, concatenating ${e}_s$ and $z_s$ into a single state vector $s_{aug}:=[e_s^\textrm{T}~ z_s^\textrm{T}]^\textrm{T}$, we derive the concatenated dynamics in the $s$-domain
\begin{align}\label{eq:error_s}
\dot{s}_\mrm{aug}&=F_\mrm{aug}(s_\mrm{aug})+G_\mrm{aug}(s_\mrm{aug})u(t),\ t\geq 0,
\end{align}
where $F_\mrm{aug}(s_\mrm{aug}):=\matt{F(e_s(t)+z_s(t))-f_d(z(t))\\f_d(z_s(t))}$ and $G_\mrm{aug}(s_\mrm{aug}):=\matt{G(e_s(t)+z_s(t))\\0}$.

\begin{remark}
To be able to transform the reference trajectory dynamics through the barrier function method, it must hold that $z(t)\in \mathcal{Q}$, for all $t\geq0$ and for all $z_i, i\in \mathcal{I}$. This assumption is quite weak since it only requires that the system tracks a trajectory that does not violate any safety constraints. \frqed
\end{remark}


\section{Optimal Tracking Sub-Problems}

To reduce the communication burden and conserve resources, the system operates under a sampled version of the transformed state
\begin{align*}
\hat{s}_\mrm{aug}(t) =  \left \{  \begin{tabular}{ll} $s_\mrm{aug}(r_j),$ & $\forall t\in(r_j,r_{j+1}]$ \\ 
			 $s_\mrm{aug}(t),$ & $t=r_{j}.$\end{tabular}  \right. 
\end{align*}
The sampling instances constitute a strictly increasing sequence $\{r_j\}_{j=0}^\infty$, where $r_j$, $j\in\mathbb{N}$, is the $j$-th consecutive sampling instant, with $r_0=0$ and $\lim\limits_{j\rightarrow \infty} r_j=\infty$.
To decide when to trigger an event, we define the triggering error as the difference between the state $s_\mrm{aug}(t)$ at the current time $t$ and the state $\hat{s}_\mrm{aug}(t)$ that was sampled most recently:
\begin{align*}
e_\mrm{trig}(t)=\hat{s}_\mrm{aug}(t)-s_\mrm{aug}(t),\ \forall t\in(r_{j},r_{j+1}],\ j\in\mathbb{N}.
\end{align*}

Our objective is to find a feedback controller that minimizes a cost functional given by
\begin{align}\nonumber
J(s_\mrm{aug}(0);u)=\frac{1}{2}\int_0^\infty e^{-\gamma \tau}\big ( s_\mrm{aug}^\textrm{T} Q_\mrm{aug} s_\mrm{aug}+R(u)   \big)\textrm{d}\tau,
\end{align}
where $\gamma\in\reals^+$ is a discount factor, and $Q_\mrm{aug}=\matt{Q& 0_{n\times n}\\ 0_{n\times n} & 0_{n\times n}}$ is a user defined matrix where $Q \succeq 0$ and $0_{n\times n}$ a square matrix of zeros.
Furthermore, to satisfy the magnitude constraint on $u$, i.e.,  $\|u\| \le \lambda$, we chose $R(u)$ to have following form, adopted from \cite{abu2005nearly},
\bea\label{eq:Thetau}
R(u) = \int_{0}^{u}\lambda \text{tanh}^{-1}\Big(\frac{v}{\lambda}\Big)\gamma_1 \textrm{d} v,~\forall u\in(-\lambda,\lambda),
\eea
where $\text{tanh}^{-1}(\cdot)$ denotes the inverse of the hyperbolic tangent function and $\gamma_1>0$. The states $s$ and control $u$ are not coupled in the safety constraints.

Initially, we consider an infinite bandwidth optimal control problem, assuming that the controller has access to the augmented transformed state at all times.
Following standard optimal control methods, we define
the optimal value function $V:\reals^{2n} \rightarrow \reals$ given, $\forall s_{\mrm{aug}}$, as
\begin{align}\label{eq:opvalue}
V(&s_\mrm{aug}(t))=\min_{u} \frac{1}{2}\int_t^\infty  e^{-\gamma(\tau-t)}\big (s_\mrm{aug}^\textrm{T} Q_\mrm{aug} s_\mrm{aug}+R(u)\big)\textrm{d}\tau,
\end{align}
and the Hamiltonian associated with \eqref{eq:error_s} and \eqref{eq:opvalue} for the continuously updating controller as follows
\begin{align}\label{eq:hamil}
{H}(s_\mrm{aug},u_c,\frac{\partial V}{\partial s_\mrm{aug}})=\frac{\partial V}{\partial s_\mrm{aug}}^\textrm{T}(F_\mrm{aug}(s_\mrm{aug})+G_\mrm{aug}(s_\mrm{aug})u_c)+ \nonumber \\+\frac{1}{2}\big(s_\mrm{aug}^\textrm{T} Q_\mrm{aug} s_\mrm{aug}+ R(u_c) -2\gamma V(s_\mrm{aug})\big),\  \forall s_\mrm{aug},u_c.
\end{align}
After employing the stationarity condition, for the Hamiltonian \eqref{eq:hamil}, i.e., $\frac{\partial {H}(\cdot)}{\partial u_c}=0$, the infinite bandwidth optimal control can be found to be
\begin{align}\label{eq:opcontrol}
u_c^\star(s_\mrm{aug})&=\arg\min_{u_c} {H}(s_\mrm{aug},u_c,\frac{\partial V}{\partial s_\mrm{aug}}) \nonumber \\
&=- \lambda \textrm{tanh}\big(\frac{1}{2\gamma_1 \lambda}G_\mrm{aug}^\textrm{T}(s_\mrm{aug})\frac{\partial V}{\partial s_\mrm{aug}}\big),\ \forall s_\mrm{aug}.
\end{align}
By substituting the optimal control \eqref{eq:opcontrol} into \eqref{eq:hamil} one has the Hamilton-Jacobi-Bellman (HJB) equation given as
\begin{align}\label{eq:hjb}
{H}(s_\mrm{aug},u_c^\star,\frac{\partial V}{\partial s_\mrm{aug}})&=0, \ \forall s_\mrm{aug}.
\end{align}

Now, to reduce the communication between  the plant and the controller we use an intermittent version of \eqref{eq:hjb} by introducing a sampled-data component that will enforce sparse and aperiodic updates for the controller. This will be designed in such a way that certain conditions on the state of the system are met, and stability is guaranteed as is shown in the analysis that follows. Thus, the controller operates with the sampled version of the system states, rather than the actual ones, and \eqref{eq:opcontrol} becomes
\begin{align}\label{eq:evopcontrol} 
u^\star(\hat{s}_\mrm{aug})=- \lambda \textrm{tanh}\big(\frac{1}{2\gamma_1 \lambda}G_\mrm{aug}^\textrm{T}(\hat{s}_\mrm{aug})\frac{\partial V}{\partial \hat{s}_\mrm{aug}}\big),\ \forall \hat{s}_\mrm{aug}.
\end{align}

\begin{assumption}\label{as:lic}
There exists a positive constant $L$ such that
\begin{align*}
&\left\|u_c(s_\mrm{aug})-u(\hat{s}_\mrm{aug})\right\| \leq L \left\|e_\mrm{trig}\right\|, ~\forall s_\mrm{aug}, \hat{s}_\mrm{aug}.
\end{align*}
\frqed
\end{assumption}

\begin{remark}
This assumption on the Lipschitz constant of the controller is standard in the literature, e.g., in \cite{girard2014dynamic, borgers2014event, mazo2011decentralized, tabuada2007event}, and is satisfied in various realistic control applications. It also enables us to show convergence of the intermittent tracking policy that evolves in the transformed $s$-space. 
\frqed
\end{remark}
\begin{theorem}
Consider the constrained system evolving in the transformed $s$-space \eqref{eq:transformation8}, following the trajectory given by \eqref{eq:exo1}. Let the augmented tracking error system be given by \eqref{eq:error_s}, and the intermittent policy by \eqref{eq:evopcontrol}. 
Then, the closed-loop error system has an asymptotically stable equilibrium when $\gamma=0$, and is ultimately uniformly bounded (UUB) when $\gamma\neq0$, under the triggering condition given by
\begin{align*}
 \left\|e_\mrm{trig}\right\|^2\leq\frac{(1/2-\beta^2)\underline{\lambda}\big(Q\big)}{L^2 \lambda \gamma_1} \left\|e_s\right\|^2+\frac{1}{L^2 \lambda \gamma_1}R(u),
 \end{align*}
 where $\beta\in(0,\frac{1}{\sqrt{2}})$ is a design parameter and $\underline{\lambda}\big(Q\big)$ is the minimum eigenvalue of $Q$.
 Furthermore, Zeno behavior is excluded via a lower bound on the inter-event times, i.e., $\exists \bar{r}>0$ such that, $r_{j+1}-r_j>\bar{r},~\forall j\in\mathcal{N}$.
\end{theorem}

\begin{proof}
Consider the time-derivative of the optimal value function \eqref{eq:opvalue}, along the trajectories of the augmented system \eqref{eq:error_s} under the intermittent policy \eqref{eq:evopcontrol} over an inter-event period, i.e., $\forall t\in (r_{j},r_{j+1}]$
\begin{align}\label{eq:lyap11}
\dot{V}=\frac{\partial V}{\partial s_\mrm{aug}}^\textrm{T}\big(F_\mrm{aug}(s_\mrm{aug})+G_\mrm{aug}(s_\mrm{aug}) u^\star(\hat{s}_\mrm{aug})\big).
\end{align}
Now, we note that the continuously triggered HJB  \eqref{eq:hjb} can be rewritten as
\begin{align*}
0 &= \frac{1}{2}s_\mrm{aug}^\textrm{T} Q_\mrm{aug} s_\mrm{aug} + \frac{\partial V}{\partial s_\mrm{aug}}^\textrm{T}\big(F_\mrm{aug}(s_\mrm{aug})) -\gamma V \\ &\quad+\lambda^2\gamma_1 \textrm{ln}\bigg(1-\textrm{tanh}^2\big(\frac{1}{2\gamma_1 \lambda}G^\textrm{T}_\mrm{aug}(s_\mrm{aug})\frac{\partial V}{\partial s_\mrm{aug}}\big)\bigg),
\end{align*}
which yields
\begin{align}\label{eq:lyap2}
\dot{V}&=-\frac{1}{2}s_\mrm{aug}^\textrm{T} Q_\mrm{aug} s_\mrm{aug} +\gamma V +\frac{\partial V}{\partial s_\mrm{aug}}^\textrm{T}(G_\mrm{aug}u^\star(\hat{s}_\mrm{aug})) \nonumber\\ &\quad-\lambda^2\gamma_1 \textrm{ln}\bigg(1-\textrm{tanh}^2\big(\frac{1}{2\gamma_1 \lambda}G^\textrm{T}_\mrm{aug}(s_\mrm{aug})\frac{\partial V}{\partial s_\mrm{aug}}\big)\bigg),
\end{align}
and, due to the structure of the optimal control function, we can get by following \cite{modares2014optimal}
\begin{align}
\dot{V}&=-\frac{1}{2} s_\mrm{aug}^\textrm{T} Q_\mrm{aug} s_\mrm{aug}+\gamma V(s_\mrm{aug})-R(u^\star(\hat{s}_\mrm{aug}) )\nonumber\\&\quad+2\gamma\lambda \int_{u^\star(s_\mrm{aug})}^{u^\star(\hat{s}_\mrm{aug})}(\frac{1}{2\gamma_1 \lambda}G_\mrm{aug}^\textrm{T}\big(\hat{s}_\mrm{aug})+\mrm{tanh}^{-1}(\frac{v}{\lambda})\big)\textrm{d}v,\nonumber
\end{align}

It has been shown that the following holds
\begin{align*}
    2\gamma\lambda \int_{u^\star(s_\mrm{aug})}^{u^\star(\hat{s}_\mrm{aug})}(\frac{1}{2\gamma_1 \lambda}G_\mrm{aug}^\textrm{T}\big(\hat{s}_\mrm{aug})+\mrm{tanh}^{-1}(\frac{v}{\lambda})\big)\textrm{d}v\\ \leq L^2 \lambda \gamma_1 \|e_\mrm{trig}\|^2.
\end{align*}

Note that the term $s_\mrm{aug}^\textrm{T} Q_\mrm{aug} s_\mrm{aug}$ is equivalent to $e_s^\textrm{T} Q e_s$ since the rest of the terms in $Q_\mrm{aug}$ are zero. Hence,
\begin{align}\label{eq:vdot1}
\dot{V}&\le-\frac{1}{2} e_s^\textrm{T} Q e_s+\gamma V(s_\mrm{aug})-R(u^\star(\hat{s}_\mrm{aug})) \\&\quad+ L^2 \lambda \gamma_1 \|e_\mrm{trig}\|^2,\nonumber
\end{align}
Therefore, for some $\beta\in(0,\frac{1}{\sqrt{2}})$, given the condition
\begin{equation}\label{eq:ub}
    L^2 \lambda \gamma_1 \|e_{\mrm{trig}}\|^2\le(\frac{1}{2}-\beta^2)\underline{\lambda}(Q)\|e_{s}\|^2+R(u^\star(\hat{s}_\mrm{aug})),
\end{equation}
the equation \eqref{eq:vdot1} becomes
\begin{equation}\nonumber
\dot{V}\le-\beta^2e_s^\textrm{T} Q e_s+\gamma V(s_\mrm{aug}).
\end{equation}
Thus, given that the discount factor is picked such that $\gamma=0$, then $\dot{V}\le-\beta^2e_s^\textrm{T} Q e_s$, and the closed-loop dynamics are asymptotically stable. We also know that $V$ is positive definite and bounded in a compact set $\Omega\subseteq \R^{2n}$ as, $\sup_{s_\mrm{aug}\in\Omega} \left\|V (s_\mrm{aug})\right\|\leq V_\mrm{max}$. Then from \eqref{eq:ub} we conclude that $\dot{V}<0$ whenever $e_s$ lies outside the compact set $\Omega_{e_s}=\big\{e_s:\left\|e_s\right\|\leq \sqrt{\frac{\gamma V_\mrm{max}}{\beta^2\underline{\lambda}(Q)}}\big\}$. Therefore, $e_s$ is UUB \cite{khalil2002nonlinear} for $\gamma\neq0$.
The exclusion of Zeno behavior follows directly from \cite{yang2020safe}.
\frQED

\end{proof}


\section{Learning Algorithm}
In this section, we employ approximation structures to solve the optimal tracking problem in a data-driven way, while guaranteeing safety constraints.
\subsection{Critic Approximator}
Initially, we employ a critic approximator, that will be able to estimate the optimal value function that solves the HJB equation. It is known that the optimal value function can be expressed as
\begin{align}\label{eq:critic}
V^\star(s_\mrm{aug})={\theta^\star_c}^\textrm{T}\phi_c(s_\mrm{aug})+\epsilon_{c}(s_\mrm{aug}), \quad\forall s_\mrm{aug}, 
\end{align}
where $\theta^\star_c\in\reals^h$ are unknown ideal weights which are bounded as $\left\|\theta^\star_c\right\|\leq \theta_\mrm{cmax}$. Furthermore, $\phi_c\eqdef[\phi_{1} \;\phi_{2} \;\cdots
\;\phi_{h}]:\reals^{2n}\rightarrow\reals^h$, is a bounded $C^1$ basis function, i.e., with bounded first order derivatives, so that $\left\|\phi_c\right\|\leq \phi_\mrm{cmax}$ and $\left\|\frac{\partial \phi_c}{\partial x}\right\|\leq\phi_\mrm{dcmax}$, and $h$ is the number of basis. Finally, $\epsilon_{c}:\reals^{2n}\rightarrow\reals$ is the approximation error.
\begin{remark}
The residual error and its derivative are bounded inside any compact set $\Omega\subset \reals^n$, i.e., $\sup\limits_{s_\mrm{aug}\in\Omega} \left\|\epsilon_c(s_\mrm{aug})\right\|\leq \epsilon_\mrm{cmax}$ and $\sup\limits_{s_\mrm{aug}\in\Omega} \left\|\frac{\partial \epsilon_c(s_\mrm{aug})}{\partial s_\mrm{aug}}\right\|\leq \epsilon_\mrm{dcmax}$. In addition, the residual error and its derivative converge uniformly to zero in $\Omega$ as $h\rightarrow\infty$, if $\phi_c$ is chosen so that it represents a complete basis of $V^\star$ inside $\Omega$ \cite{hornik1990universal}. \frqed
\end{remark}
Based on this, the optimal intermittent policy in
can be re-written, $\forall t\in (r_{j},r_{j+1}]$, as
\begin{align}\label{eq:starvalue1}
u^\star(\hat{s}_\mrm{aug})&=-\lambda \textrm{tanh}\big(\frac{1}{2\gamma_1 \lambda}G_\mrm{aug}^\textrm{T}(\hat{s}_\mrm{aug})\times\\ &\times\big(\frac{\partial \phi(\hat{s}_\mrm{aug})}{\partial \hat{s}_\mrm{aug}}^\textrm{T} \theta^\star_c
+\frac{\partial \epsilon_c(\hat{s}_\mrm{aug})}{\partial \hat{s}_\mrm{aug}}\big)\big). \nonumber
\end{align}
\subsection{Actor Approximator}
We employ another approximating structure, called an \emph{actor}, to approximate the intermittent controller \eqref{eq:starvalue1}. This is expressed, $\forall t\in(r_{j},r_{j+1}]$, as
\begin{align}\label{eq:actor}
u^\star(\hat{s}_\mrm{aug})={\theta_u^\star}^\textrm{T} \phi_u(\hat{s}_\mrm{aug})+\epsilon_u(\hat{s}_\mrm{aug}), \ \forall \hat{s}_\mrm{aug},
\end{align}
where $\theta_u^\star\in\reals^{h_2\times m}$ are the optimal weights, $\phi_u$ are the basis functions defined similarly to the critic approximator, $h_2$ is the number of basis, and $\epsilon_u$ is the actor approximation error.
\begin{remark}
The approximation error $\epsilon_u$ is bounded over a compact set $\Omega\subset \reals^n$, so that $\sup\limits_{\hat{s}_\mrm{aug}\in\Omega}\left\|\epsilon_u\right\|\leq \epsilon_\mrm{umax}$. In addition, the approximation basis functions are chosen so that they are upper bounded by positive constants, i.e, $\left\|\phi_u\right\|\leq \phi_\mrm{umax}$. \frqed
\end{remark}
Thus, the current estimates of the value function and the optimal policy are derived based on estimations of the ideal critic and actor weights, denoted $\hat{\theta}_c$ and $\hat{\theta}_u$, respectively
\begin{align}\label{eq:apcritic}
\hat{V}(s_\mrm{aug}(t))=\hat{\theta}_c^\textrm{T}\phi_c(s_\mrm{aug}(t)), \forall s_\mrm{aug},
\end{align}
\begin{align}\label{eq:apactor}
\hat{u}(\hat{s}_\mrm{aug})=\hat{\theta}_u^\textrm{T}\phi_u(\hat{s}_\mrm{aug}), \forall \hat{s}_\mrm{aug}.
\end{align}
Now, the learning mechanism comprises tuning laws that will allow us to obtain optimal estimates to the critic and actor weights. Towards this, we define the estimation error $e_c\in\reals$ based on the Hamiltonian function as
\begin{align*}
e_c&=H(s_\mrm{aug},\hat{u}(\hat{s}_\mrm{aug}),\frac{\partial \hat{V}(s_\mrm{aug})}{\partial s_\mrm{aug}})\\&\quad-H(s_\mrm{aug},u_c^\star(s_\mrm{aug}),\frac{\partial V^\star(s_\mrm{aug})}{\partial s_\mrm{aug}})\nonumber\\&
=\hat{\theta}_c^\textrm{T} \frac{\partial \phi_c}{\partial s_\mrm{aug}}\big(F_\mrm{aug}(s_\mrm{aug})+G_\mrm{aug}(s_\mrm{aug})\hat{u}(\hat{s}_\mrm{aug})\big)-\gamma\hat{\theta}_c^\textrm{T} \phi_c+\hat{r}\\&
=\hat{\theta}_c^\textrm{T} \omega+\hat{r},
\end{align*}
with $\omega=\frac{\partial \phi_c}{\partial s_\mrm{aug}}\big(F_\mrm{aug}(s_\mrm{aug})+G_\mrm{aug}(s_\mrm{aug})\hat{u}(\hat{s}_\mrm{aug})\big)-\gamma \phi_c$, $\hat{r}=\frac{1}{2}s_\mrm{aug}^\textrm{T} Q_\mrm{aug} s_\mrm{aug}+R(\hat{u}(\hat{s}_\mrm{aug}))$ and $H(s_\mrm{aug},u_c^\star(s_\mrm{aug}),\frac{\partial V^\star(s_\mrm{aug})}{\partial s_\mrm{aug}})=0$ from \eqref{eq:hjb}, and for simplicity we omit the dependence of $\phi_c$ on the augmented $s$-state. 
\par
To drive the error $e_c$ to zero one has to pick the critic weights appropriately. By defining the squared-norm error as  $E_c=\frac{1}{2} e_c^2$ we can apply the normalized gradient descent method to obtain the estimate of the critic weights as
{\small
\begin{align}\label{eq:rls}
\dot{\hat{\theta}}_c=-\alpha\frac{1}{(\omega^\textrm{T}\omega +1)^2} \frac{\partial E_c}{\partial \hat{\theta}_c}= -\alpha\frac{\  \omega}{(\omega^\textrm{T}\omega +1)^2}e_c,
\end{align}}
where $\alpha\in\reals^+$ is a tuning parameter.

By defining the critic error dynamics as $\tilde{\theta}_c=\theta^\star_c-\hat{\theta}_c$ and taking its derivative with respect to time one has
\begin{align}\label{eq:criticerror}
\dot{\tilde{\theta}}_c=-\alpha \frac{\omega\ \omega^\textrm{T}}{(\omega^\textrm{T}\omega +1)^2}\tilde{\theta}_c+\alpha\frac{\omega}{(\omega^\textrm{T}\omega +1)^2}\epsilon_\mrm{Hc},
\end{align}
where $\epsilon_\mrm{Hc}=-\frac{\partial \epsilon_c}{\partial x}(F+G \hat{u}), \forall x, \hat{u}$, is upper bounded by $\epsilon_\mrm{Hcmax}\in\reals^+$ as $\left\|\epsilon_{Hc}\right\|\leq\epsilon_\mrm{Hcmax}$.

To state stability results on the derived learning system, we can consider the critic error dynamics as a sum of nominal dynamical behavior with a time dependent perturbation due to the approximation error, denoted respectively as $S_\mrm{N}$ and $S_\mrm{P}$, where $\dot{\tilde{\theta}}_c=S_\mrm{N}+S_\mrm{P}$ with $S_\mrm{N}=-\alpha \frac{\omega\ \omega^\textrm{T}}{(\omega^\textrm{T}\omega +1)^2}\tilde{\theta}_c$ and $S_\mrm{P}=\alpha\frac{\omega}{(\omega^\textrm{T}\omega +1)^2}\epsilon_\mrm{Hc}$.

\begin{theorem}\label{th:KhalilTheorem}
Assume that the signal $M=\frac{\omega}{(\omega^\textrm{T}\omega +1)}$ is persistently exciting, i.e., $\int_t^{t+T}M M^\textrm{T} d\tau\geq b I$, $\forall t\ge 0$ for some $b,~T\in\reals^+$, where $I$ is an identity matrix of appropriate dimensions, and also that there exists $M_b\in\reals^+$ such that for all $t\geq 0$, $\max\big\{\left|M\right|,\left|\dot{M}\right|\big\}\leq M_B$. Then, the nominal system $S_\mrm{N}$ is exponentially stable and its trajectories satisfy $\left\|\tilde{\theta}_c(t)\right\|\leq \left\|\tilde{\theta}_c(0)\right\|\kappa_1 e^{- \kappa_2t}$, for some $\kappa_1,\ \kappa_2\in\reals^+$ and for all $t\ge0$.
\end{theorem}
\begin{proof}
The proof follows from \cite{vamvoudakis2017event}. \frQED
\end{proof}
Now, to derive the tuning laws for the actor approximator, we consider the error $e_u\in\reals^m$ given, $\forall \hat{s}_\mrm{aug}$, by
\begin{align*}
e_u 	&=\hat{u}-u_{\hat{\theta}_c}\\
	&=\hat{\theta}^\textrm{T}_u \phi_u(\hat{s}_\mrm{aug})+ \lambda \textrm{tanh}\big(\frac{1}{2\gamma_1 \lambda}G_\mrm{aug}^\textrm{T}(\hat{s}_\mrm{aug}) \big(\frac{\partial \phi(\hat{s}_\mrm{aug})}{\partial \hat{s}_\mrm{aug}}^\textrm{T} \theta^\star_c\big),
\end{align*}
where $u_{\hat{\theta}_c}$ is the controller based on the critic weights $\hat{\theta}_c$.

The objective is to select $\hat{\theta}_u$ such that the error $e_u$ goes to zero. As such, we minimize the following error cost
\begin{align*}
E_u=\frac{1}{2} e_u^\textrm{T} e_u.
\end{align*}
In keeping with our objective to avoid over-utilization of the system's resources, we let the actor learn in an aperiodic fashion, by updating the weights only at the triggering instances, and keeping them constant between them. This gives the system an impulsive nature, whose behavior we investigate based on results of \cite{haddad2006impulsive} and \cite{hespanha2008lyapunov}.

Then, the update laws are given by
\begin{align}\label{eq:cactor}
\dot{\hat{\theta}}_u(t)=0,\ \forall t\in\reals^+ \setminus \underset{{j\in\mathbb{N}}}{\bigcup} r_j ,
\end{align}
and the jump equation to compute $\hat{\theta}_u(r_j^+)$ given, $\text{for } t=r_j$, by
\begin{align}\label{eq:dactor}
&{\hat{\theta}}_u^+=\hat{\theta}_u-\alpha_u \phi_u({x}_\mrm{aug}(t))\bigg(\hat{\theta}_u^\textrm{T} \phi_u(s_\mrm{aug}(t))\nonumber\\&+\lambda \textrm{tanh}\bigg(\frac{1}{2\gamma_1 \lambda}G_\mrm{aug}^\textrm{T}(\hat{s}_\mrm{aug}) \frac{\partial \phi(\hat{s}_\mrm{aug})}{\partial \hat{s}_\mrm{aug}}^\textrm{T} \theta^\star_c\bigg)^\textrm{T}.
\end{align}
 By defining the actor error dynamics as $\tilde{\theta}_u=\theta_u^\star-\hat{\theta}_u$ and taking the time derivative using the continuous update \eqref{eq:cactor} and by using the jump system \eqref{eq:dactor} updated at the trigger instants one has
\begin{align}\label{eq:ecactor}
\dot{\tilde{\theta}}_u(t)=0,\ \forall t\in\reals^+ \setminus \underset{{j\in\mathbb{N}}}{\bigcup} r_j ,
\end{align}
and, $\text{for } t=r_j$
\begin{align}\label{eq:edactor}
&{\tilde{\theta}}_u^+=\tilde{\theta}_u-\alpha_u \phi_u(s_\mrm{aug}(t)) \phi_u(s_\mrm{aug}(t))^\textrm{T} \tilde{\theta}_u(t)\nonumber\\&
\\&-\lambda \textrm{tanh}\big(\frac{1}{2\gamma_1 \lambda}G_\mrm{aug}^\textrm{T}(\hat{s}_\mrm{aug}) \big(\frac{\partial \phi(\hat{s}_\mrm{aug})}{\partial \hat{s}_\mrm{aug}}^\textrm{T}+\frac{\partial \epsilon_c(\hat{s}_\mrm{aug})}{\partial \hat{s}_\mrm{aug}}^\textrm{T}\big) \theta^\star_c\big),
\end{align}
respectively.
Note that the solution of \eqref{eq:ecactor}-\eqref{eq:edactor} is left continuous; that is, it is continuous everywhere except at the resetting times $r_j$
\begin{align*}
&\hat{\theta}_u(r_j)=\lim _{\delta\rightarrow 0^+}  \hat{\theta}_u(r_j-\delta), ~ \forall j \in \mathbb{N},
\end{align*}
and
\begin{align*}
{\hat{\theta}}_u^+&=\hat{\theta}_u-\alpha_u \phi_u({x}_\mrm{aug}(t))\bigg(\hat{\theta}_u^\textrm{T} \phi_u(s_\mrm{aug}(t))\\&\qquad+ \lambda \textrm{tanh}\big(\frac{1}{2\gamma_1 \lambda}G_\mrm{aug}^\textrm{T}(\hat{s}_\mrm{aug}) \frac{\partial \phi(\hat{s}_\mrm{aug})}{\partial \hat{s}_\mrm{aug}}^\textrm{T} \theta^\star_c\bigg)^\textrm{T},\ t=r_j.
\end{align*}

\section{Simulation Results}
To validate the effectiveness of the proposed framework, we solve a safety-critical task described as a temporal logic specification, which can be decomposed as a sequence of event-triggered optimal tracking control problems. Same as the example in Sec. \ref{sec:dec}, we consider the LTL specification $\phi=\phi_{\textrm{c}}\wedge\phi_{\textrm{s}}$, with $\phi_{\textrm{c}}=\Diamond p_2\  \wedge \ (\lnot p_2 \mathcal{U} p_1)$ and $\phi_{\textrm{s}}=\Box p_3$, where $p_1$, $p_2$ are to track different target trajectories and $p_3$ is a safe zone to stay in. Given the FSA constructed based on $\phi_{\textrm{c}}$ as shown in Fig. \ref{f:automaton}, Problem \ref{Problem:safety} can be decomposed into two sub-problems (Problem \ref{Problem:safety2}). Due to space limitations, we only show the results of the first sub-problem (i.e., reaching $p_1$), and the results for the second sub-problems are omitted as they can be obtained in a similar manner. The safety constraint $p_3$ is defined as $\mathcal{Q}= \{x\in\reals^n|c \le Ax + r \le C\}$, where $A = I$ (the identity matrix), $r = \mathbf{0}_{4\times 1}$ and $c = -30\times \mathbf{1}_{4\times 1}, C = -c$.

We set the predicate $p_1=(\vert\vert x(t)-z_1(t)\vert\vert\le \epsilon)$, where the trajectory to be tracked is $z = 0.5\times[\sin{0.5t}, \cos{0.5t}]^\T$, and $\epsilon=0.6$. The system dynamics are given as $\dot{x} = f(x) + g(x)u$, where,
\begin{equation*}f(x) = \begin{bmatrix}
-x_1+x_2\\
-0.5x_1-0.5x_2(1-(\cos{2x_1}+2)^2
\end{bmatrix},
\end{equation*}
and
\begin{equation*}g(x) = \begin{bmatrix}
0\\
\cos{2x_1}+2
\end{bmatrix}.
\end{equation*}

For the learning algorithm, the initial actor and critic weights are picked randomly in $[0, 1]$. The user-defined parameters are selected as $Q = 800I$. The evolution of the tracking error is shown in Fig. \ref{f:error}. In Fig.\ref{f:control} we present the control input of the system after the exploration noise has been sufficiently decreased. It can be seen that due to the nature of the desired trajectory, the controller does not achieve perfect tracking, but rather renders the tracking error UUB. It is validated that the safety constrains are satisfied while the tracking error is decreasing and eventually bounded as proved.

\begin{figure}[!ht]
	\vspace{-0.1cm}
	\hspace{1cm}
	\centering
	\includegraphics[width=\linewidth]{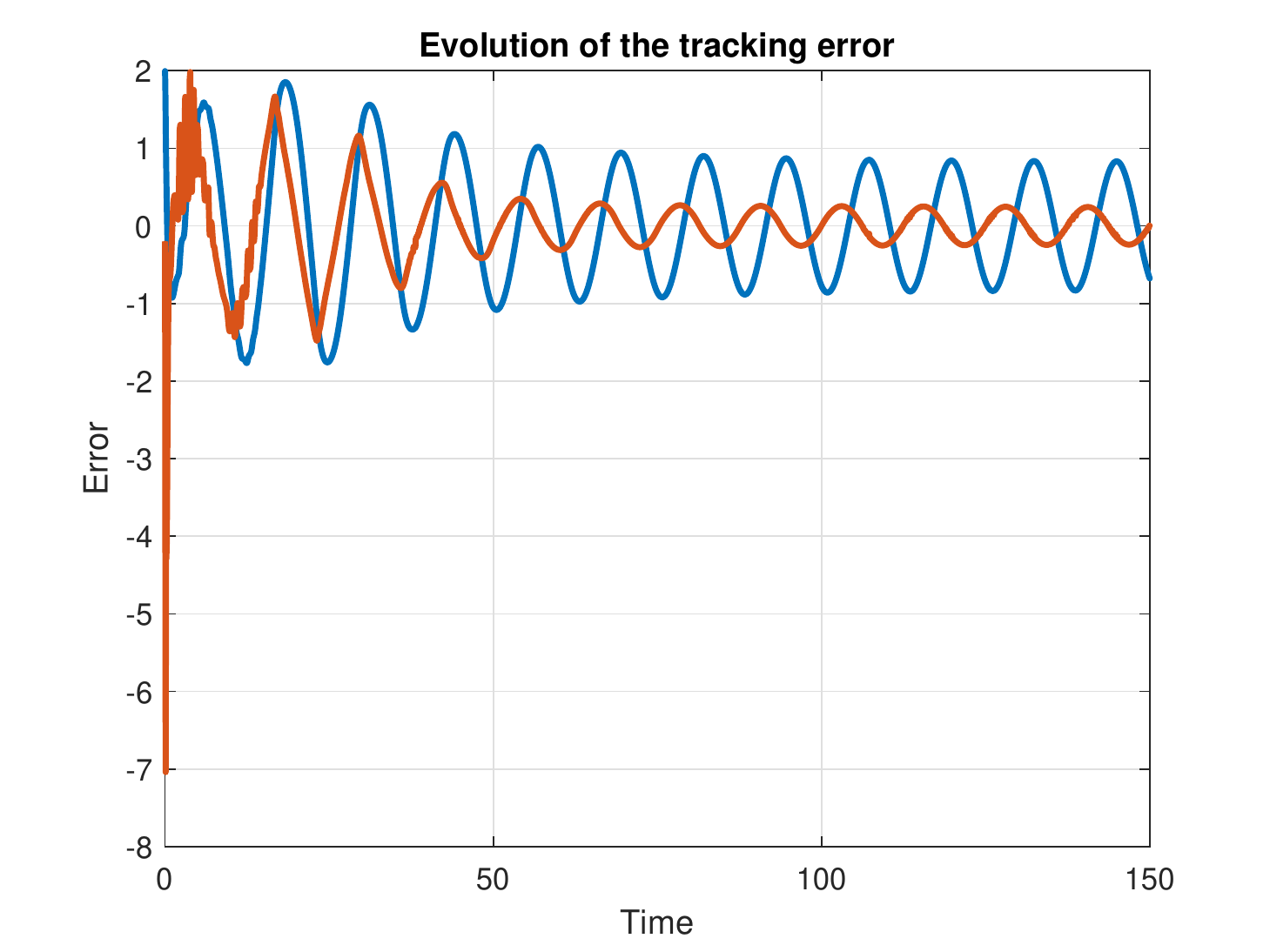}
	\vspace{-0.2cm}
	\caption{The evolution of the tracking error of the states.}
	\label{f:error}
\end{figure}

\begin{figure}[!ht]
	\vspace{-0.1cm}
	\hspace{1cm}
	\centering
	\includegraphics[width=\linewidth]{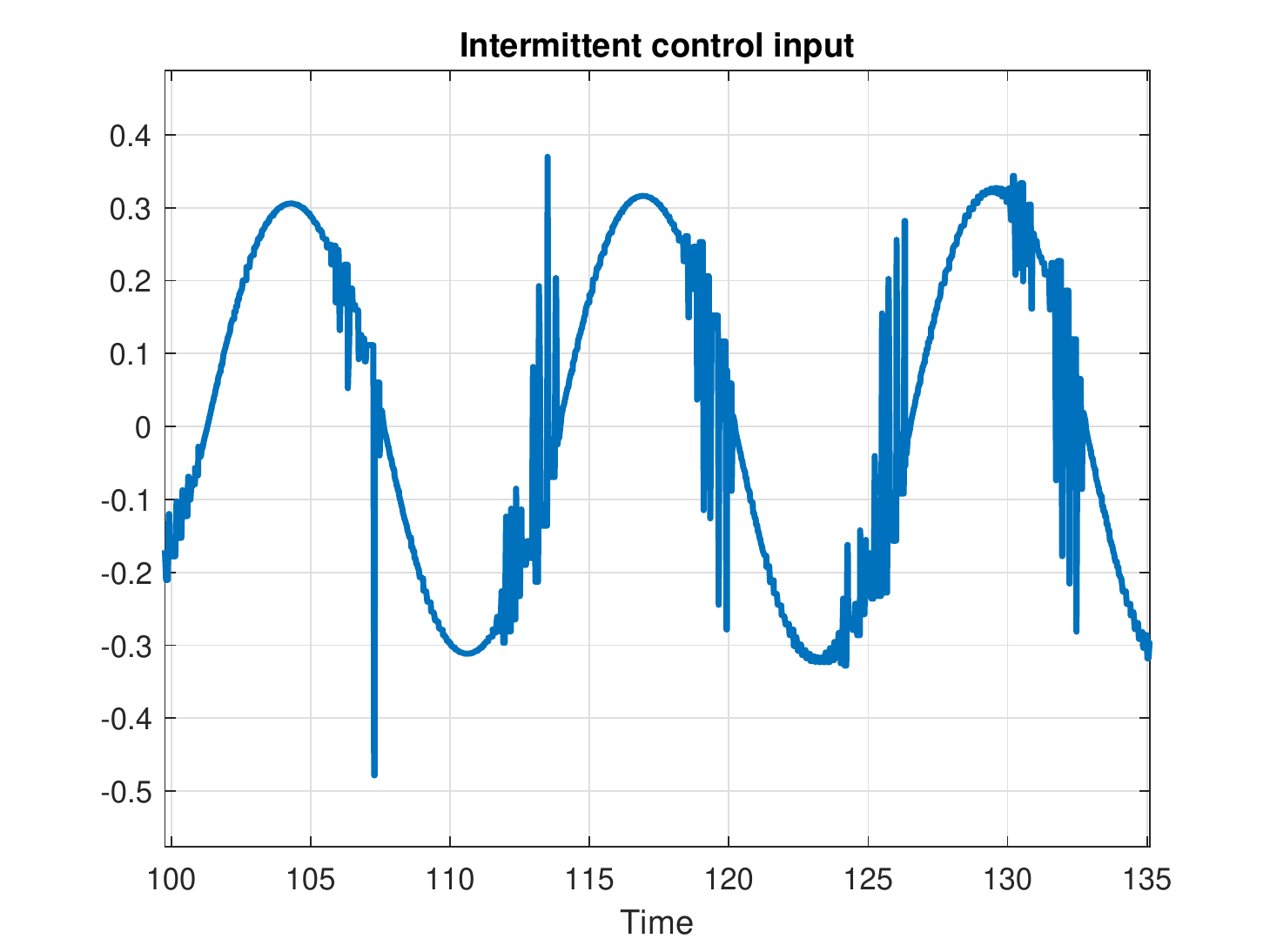}
	\vspace{-0.2cm}
	\caption{Intermittent control policy.}
	\label{f:control}
\end{figure}

\section{Conclusion and Future Work}
In this paper, we developed an intermittent learning framework for a system tasked with a complex mission while guaranteeing safety. We brought together ideas from LTL and control-oriented RL to decompose the mission into a sequence of tracking sub-problems which are constrained by safety specifications. We convert the system using barrier functions, thus, deriving an unconstrained optimal tracking problem in the transformed state space. The tracking problem was tackled via the construction of intermittent policies and guarantees of stability and optimality were presented. Finally, to circumvent the issues arising from the difficulty of solving the underlying HJB equations, we used an RL algorithm to obtain estimated versions of the intermittent safe optimal control policies.

Future work will focus on extensions to the settings of distributed multi-agent systems.

\bibliographystyle{IEEEtran}
\bibliography{reference}


\begin{comment}
\documentclass[conference]{IEEEtran}
\IEEEoverridecommandlockouts
\usepackage{cite}
\usepackage{amsmath,amssymb,amsfonts}
\usepackage{algorithmic}
\usepackage{graphicx}
\usepackage{textcomp}
\usepackage{xcolor}

\usepackage{graphicx}
\usepackage{epsfig}
\usepackage{color}
\usepackage{enumerate}

\newcommand{\bm}[1]{\mbox{$\boldmath{#1}$}}
\newcommand{\mc}[1]{\mathcal{#1}}
\newcommand{\field}[1]{{\mathbb{#1}}}
\newcommand{\bea}{\begin{eqnarray}}
\newcommand{\eea}{\end{eqnarray}}
\newcommand{\beas}{\begin{eqnarray*}}
\newcommand{\eeas}{\end{eqnarray*}}
\newcommand{\leftm}{\left[\begin{array}}
\newcommand{\rightm}{\end{array}\right]}
\newcommand{\reals}{\mbox{$\mathbb R$}}
\newcommand{\complex}{\mbox{$\mathbb C$}}
\newcommand{\ones}{\textbf{1}}
\newcommand{\zeros}{\textbf{0}}
\newcommand{\mA}{\mathcal{A}}
\newcommand{\mI}{\mathcal{I}}
\newcommand{\mC}{\mathcal{C}}
\newcommand{\mL}{\mathcal{L}}
\newcommand{\mV}{\mathcal{V}}
\newcommand{\mN}{\mathcal{N}}
\newcommand{\blue}[1]{{\color{blue}#1}}
\newcommand{\red}[1]{{\color{red}#1}}
\newcommand{\xa}{x_{\text{aug}}}
\DeclareMathOperator*{\argmax}{arg\,max}
\DeclareMathOperator*{\argmin}{arg\,min}
\DeclareMathOperator*{\argsup}{arg\,sup}
\DeclareMathOperator*{\softmax}{softmax}
\newcommand{\T}{\textrm{T}}
\newcommand{\dt}{\textrm{d}}

\newtheorem{thm}{Theorem}
\newtheorem{cor}[thm]{Corollary}
\newtheorem{theo}[thm]{Theorem}
\newtheorem{prop}[thm]{Proposition}
\newtheorem{lem}[thm]{Lemma}
\newtheorem{definition}[thm]{Definition}
\newtheorem{rem}[thm]{Remark}
\newtheorem{Assumption}[thm]{Assumption}
\newtheorem{Lemma}[thm]{Lemma}
\newtheorem{Problem}[thm]{Problem}
\newtheorem{Example}[thm]{Example}
\newtheorem{remark}[thm]{Remark}

\def\BibTeX{{\rm B\kern-.05em{\sc i\kern-.025em b}\kern-.08em
    T\kern-.1667em\lower.7ex\hbox{E}\kern-.125emX}}
\begin{document}

\title{Temporal logic constrained optimal tracking control with Guaranteed Safety\\
\thanks{Identify applicable funding agency here. If none, delete this.}
}

\author{\IEEEauthorblockN{1\textsuperscript{st} Given Name Surname}
\IEEEauthorblockA{\textit{dept. name of organization (of Aff.)} \\
\textit{name of organization (of Aff.)}\\
City, Country \\
email address or ORCID}
\and
\IEEEauthorblockN{2\textsuperscript{nd} Given Name Surname}
\IEEEauthorblockA{\textit{dept. name of organization (of Aff.)} \\
\textit{name of organization (of Aff.)}\\
City, Country \\
email address or ORCID}
\and
\IEEEauthorblockN{3\textsuperscript{rd} Given Name Surname}
\IEEEauthorblockA{\textit{dept. name of organization (of Aff.)} \\
\textit{name of organization (of Aff.)}\\
City, Country \\
email address or ORCID}
\and
\IEEEauthorblockN{4\textsuperscript{th} Given Name Surname}
\IEEEauthorblockA{\textit{dept. name of organization (of Aff.)} \\
\textit{name of organization (of Aff.)}\\
City, Country \\
email address or ORCID}
\and
\IEEEauthorblockN{5\textsuperscript{th} Given Name Surname}
\IEEEauthorblockA{\textit{dept. name of organization (of Aff.)} \\
\textit{name of organization (of Aff.)}\\
City, Country \\
email address or ORCID}
\and
\IEEEauthorblockN{6\textsuperscript{th} Given Name Surname}
\IEEEauthorblockA{\textit{dept. name of organization (of Aff.)} \\
\textit{name of organization (of Aff.)}\\
City, Country \\
email address or ORCID}
}

\maketitle

\begin{abstract}
\end{abstract}

\begin{IEEEkeywords}
\end{IEEEkeywords}

\section{Introduction}

\section{Problem Statement}
We consider the following time-invariant control-affine nonlinear system

\bea\label{eq:system}
&&\dot{x}_i = f_i(\bar{x}_i) + g_i(\bar{x}_i)x_{i+1}, \forall i = 1,\ldots,n-1, \nonumber\\
&&\dot{x}_n = f_n(\bar{x}_n) + g_n(\bar{x}_n)u\nonumber\\
&&x(0) = x_0,\ t\geq 0
\eea
where $x = [x_1,\ldots,x_n]^\textrm{T} \in\reals^n, u\in\reals$, are the states, and the control inputs, respectively, and $f_i(x):\reals^n\to\reals, g_i(x):\reals^n\to\reals$ , $\forall i=1,\dots,n$. Moreover, $\bar{x}_i = [x_1,\ldots,x_i]^\textrm{T}$ is a subset of the states. For simplicity, we neglect the dependence of the variables on the time $t$ and \eqref{eq:system} is further denoted as $\dot{x} = f(x) + g(x)u$.

System \eqref{eq:system} is expect to achieve certain goals constrained by temporal logic specifications (which we will describe later) via following one of the trajectories below
\bea\label{eq:system-tracked}
\dot{z}_i = h_i(z), z_i(0) = z_{i,0}, \forall t \ge 0, i \in \mI,
\eea
where $z_i\in\reals^n$ is the $i$-th candidate of the desired trajectories to be tracked, $h_i:\reals^n \to \reals^n$ is a Lipschitz continuous function with $h_i(0) = 0$ and $\mI$ is the set of the trajectories to be tracked.
As we aim to track a certain trajectory, we shall define the tracking error as follows
\bea\label{eq:error}
e := x - z
\eea
where $z\in z_{\mI}$ is one of the trajectory candidate to follow. Then the dynamics of $e$ is
\bea\label{eq:e_dynamics}
\dot{e} = f(e+z) + g(e+z)u - h(z).
\eea
Furthermore, with $\xa = [e^\T, z^\T]^\T$ as the augmented states, the new system can be rewritten in the following  by combining \eqref{eq:system-tracked} and \eqref{eq:e_dynamics} 
\bea\label{eq:augment}
\dot{x}_{aug} = \bar{f}(\xa) + \bar{g}(\xa)u,
\eea
where $\bar{f}(\xa) = \begin{bmatrix}
f(e+z) - h(z) \\
h(z)
\end{bmatrix}$
and $\bar{g}(\xa) = \begin{bmatrix}
g(e+z) \\
0
\end{bmatrix}$.

\subsection{Temporal Logic Syntax and Semantics}

With predicates $p_1$ and $p_2$, a linear temporal logic (LTL) specification has the following syntax, \textcolor{blue}{Zhe: I suggest that we use = instead of $\models$ as $\models$ usually means "satisfies" in the literature, and use $\phi$ to denote formulas, $p$ to denote predicates only.}
\bea\label{eq:TL_syntax}
&&p \models \top | \bigcirc p_1 | \Box p_1 | \lnot p_1 | \Diamond p_1 | p_1 \wedge p_2 | p_1 \vee p_2 |  \nonumber\\
&&\hspace{2cm}p_1 \Rightarrow p_2 | p_1 \mathcal{U} p_2 | p_1 \mathcal{T} p_2,
\eea
where $\top$  is the true Boolean constant, $\bigcirc$ (next), $\Box$ (always), $\Diamond$ (eventually), $\mathcal{U}$ (until), $\mathcal{T}$ (then) are the temporal operators, and $\lnot$(negation/not), $\wedge$ (conjunction/and), and $\vee$ (disjunction/or) are the Boolean connectives. Consider an example formula of the form $p \models \Diamond p_1\  \wedge \ \Diamond p_2 \  \wedge\  \Box p_3$ which denotes that the trajectory of the state will eventually satisfy $p_1$ and $p_2$ while always satisfying $p_3$. For example, one tracking pattern can be like i) A unmanned aerial vehicle (UAV) monitors the area in a special zig-zag; ii) If the battery is low, the UAV will track and try to reach the charging station; iii) If the fire is found, the UAV will try to approach the monitoring station to be within and communication range and thus activate the alert~\cite{vamvoudakis2017event}. 

\begin{Problem}\label{Problem:safety}
Define $\lambda >0$ and $Q$ a convex polytope. For the system given in \eqref{eq:system}, find a control policy such that the closed-loop system has a stable equilibrium point, while the control input satisfies $\|u\| \le \lambda$, and the states $x\in Q$ tracks a sequence of trajectories specified via a LTL constraint given by $p$.\frqed
\end{Problem}

\subsection{Decomposition of LTL}
Given that the TL specification is related to the time evolution, the concrete safety constraint in a given time interval varies. As a result, via the use of FSA, we manually divided Problem \ref{Problem:safety} into a series of sequential sub-problems, that each satisfies a certain safety constraint. Consequently, the original TL safety constraint is eventually satisfied. Consider now the LTL $p \models \Diamond p_2\  \wedge \ (\lnot p_2 \mathcal{U} p_1) \  \wedge\  \Box \lnot p_3$, whose FSA is shown in Figure~\ref{f:automaton}. Given that $p_1$ and $p_2$ are disjoint, the only path from the initial FSA state T0\_init to the final state accept\_S0 is: T0\_init $\to$ T0\_S2 $\to$ accept\_S0, where the state accept\_S0 means the LTL $p$ is satisfied. Thus, there are two TPBVP sub-problems, with different safety constraints. For the first one (T0\_init $\to$ T0\_S2), the boundary conditions are the initial states and $p_1$ while the safety zone is outside of $p_2 \bigcup p_3$. Similarly, for the second sub-problem (T0\_S2 $\to$ accept\_S0), the boundary conditions are $p_1$ and $p_2$ while the safety zone is outside of $p_3$. However, systematically for an co-safe LTL formula with a complex FSA, it is not straightforward to get the path from an initial FSA state to a final FSA state. Hence, we can view the FSA as a directed graph and approximately estimate the edge weights, i.e., distance, and then cast it as a shortest path problem. So far, we have finished decomposing complex LTL specifications into sequential sub-problems, that as a whole will eventually satisfy the original co-safe LTL constraint.

\begin{figure}[!ht]
	\vspace{-0.1cm}
	\hspace{0.5cm}
	\includegraphics[scale=0.6]{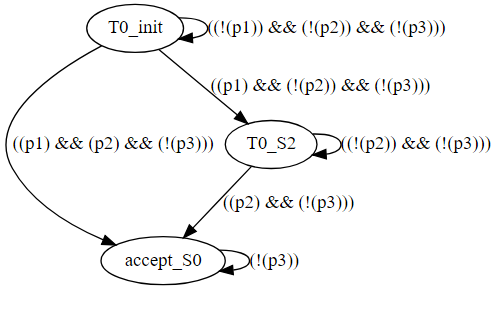}
	\vspace{-0.5cm}
	\caption{Finite state automaton generated by co-safe LTL formula $p \models \Diamond p_2\  \wedge \ (\lnot p_2 \mathcal{U} p_1) \  \wedge\  \Box \lnot p_3$.}
	\label{f:automaton}
\end{figure}

\subsection{Sub-Problem tracking a certain trajectory}
In a sub-problem, a certain trajectory $z_i, i\in \mI$ is tracked. Moreover, the state is subject to some safety constraints as $x\in Q$ and $Q = \{x|c \le Ax + r \le C\}$ with $A = [a_1^\T;a_2^\T;\ldots;a_m^\T]^\T\in\reals^{m\times n}$, $r = [r_1,\ldots,r_m]$, $c = [c_1,\ldots,c_m]$, $C = [C_1,\ldots,C_m] \in\reals^m$ and $q_i = a_ix + r_i$.

\begin{Problem}\label{Problem:safety2}
(sub-problem) For the system given by \eqref{eq:system}, find a control policy such that the system tracks a certain trajectory $i\in\mI$ given an initial condition $x_0$, while the control input satisfies
$\|u\| \le \lambda$, and the states satisfy $x\in Q$.\frqed
\end{Problem}
As a sub-problem of Problem~\ref{Problem:safety}, the Problem~\ref{Problem:safety2} translates a temporal logic constraint $p$ to a more detailed and time-invariant constraint in the current time interval.

To deal with the safety constraint, a barrier function is defined by,
\bea\label{eq:barrier}
b(q,{c}_0,{C}_0) = \log\Big(\frac{C_0}{c_0}\frac{c_0-q}{C_0-q}\Big),\forall p\in(c_0,C_0),
\eea
where ${c}_0 < 0 < {C}_0$. This does not lose any generality since we can always satisfy this via adjusting $r_i$. Moreover, $b(q,{c}_0, {C}_0)$ is invertible in the interval $({c}_0, {C}_0)$ as follows
\bea\label{eq:barrier_inverse}
b^{-1}(y,{c}_0, {C}_0) = c_0C_0\frac{e^{\frac{y}{2}} - e^{-\frac{y}{2}}}{c_0e^{\frac{y}{2}} - C_0e^{-\frac{y}{2}}},\forall y\in \reals,
\eea
with dynamics,
\bea\label{eq:barrier_inverse_derivatives}
\frac{\dt b^{-1}(y,{c}_0, {C}_0)}{\dt y} = \frac{C_0c_0^2 - c_0C_0^2}{c_0^2e^y - 2c_0C_0 + C_0^2e^{-y}}.
\eea

A system transformation of \eqref{eq:system} that accounts for safety can be written as,
\bea\label{eq:transformation}
s_i &=& b(q_i(x),c_i,C_i) \nonumber\\
q_i(x) &=& b^{-1}(s_i,c_i,C_i) \\
q_i(x) &=& a_ix + r_i, \forall i=1,\ldots,m. \nonumber
\eea

Through the use of the chain rule, we have that,
\bea
\frac{\dt q_i(x)}{\dt t} = a_i\dot{x} = \frac{\dt b^{-1}(s_i,c_i,C_i)}{\dt s_i}\frac{\dt s_i}{\dt t},\nonumber
\eea
which yields,
\bea\label{eq:transformation2}
\frac{\dt s_i}{\dt t} = \frac{1}{\frac{\dt b^{-1}(s_i,c_i,C_i)}{\dt s_i}}a_i^\textrm{T}\dot{x}.
\eea

We can thus write in a compact form,
\bea\label{eq:system2}
&&\dot{x} = f(x) + g(x)u\nonumber\\
&&x(0) = x_0,
\eea
with $f(x) = [f_i(\bar{x}_i) + g_i(\bar{x}_i)x_{i+1}, \ldots, f_{n-1}(\bar{x}_{n-1}) + g_{n-1}(\bar{x}_{n-1})x_{n},f_n(\bar{x}_n)]^\T$ and $g(x) = [0,\ldots,0,g_n(\bar{x}_n)]$. Additionally, from equation \eqref{eq:transformation}, we have that,
\bea\label{eq:transformation5}
Ax + r = b^{-1}(s,c,C),
\eea
with $b^{-1}(s,c,C) = [b^{-1}(s_1,c_1,C_1),\ldots,b^{-1}(s_m,c_m,C_m)]^\textrm{T} \in\reals^m$.

It also follows that,
\bea\label{eq:transformation6}
x  = (A^\T A)^{-1}A^\T(b^{-1}(s,c,C) - r).
\eea

\begin{remark}
It is worth noting that to $A^\T A\in\mathcal{S}^{n}$ invertible, it will be necessary for $A$ to be full column rank, i.e., $m \ge n$. That is not restricted since we can always add sufficiently large trivial bounds on the states, such as $-M \le x_1  \le M$ with $M>0$.
\end{remark}

 Now combining \eqref{eq:transformation2}, \eqref{eq:system2}, and \eqref{eq:transformation6}, one has,
\bea\label{eq:transformation7}
\frac{\dt s_i}{\dt t} &=& \frac{1}{\frac{\dt b^{-1}(s_i,c_i,C_i)}{\dt s_i}}a_i\dot{x} \\
&=&\frac{1}{\frac{\dt b^{-1}(s_i,c_i,C_i)}{\dt s_i}}a_i \big(f(x) + g(x)u \big),\nonumber
\eea
with $x = (A^\textrm{T}A)^{-1}A^\textrm{T}(b^{-1}(s,c,C) - r)$. For notational convenience, the system \eqref{eq:transformation7} will be rewritten in a compact form as
\bea\label{eq:transformation8}
\dot{s} = F(s) + G(s)u.
\eea

In Problem~\ref{Problem:safety2}, the terminal condition is required that the tracking error be zero. Correspondingly, in \eqref{eq:transformation8} the terminal condition is $s \to s(e=0)$. Then a new equivalent problem is defined as follows.

\begin{Problem}\label{Problem:safety3}
For the system \eqref{eq:transformation8}, find a policy $u$ such that the system reaches the terminal condition $e\to 0$ from the current condition such that the performance
\[J(s_0) = \int_{0}^{\infty}\big(\Psi(s) + \Theta(u) \big)\dt t\]
is minimized by $u$ subject to the dynamics given by \eqref{eq:transformation8}, $\|u\| \le \lambda > 0$,
where $\Psi(s)$ is positive definite and monotonically increasing with regards to $\|s\|$, and $\Theta(u), $ is a positive definite integrand function. For notational simplicity, we shall define $U(s,u,v):= \Psi(s) + \Theta(u)$.\frqed
\end{Problem}

To satisfy the safety constraint on $u$, i.e.,  $\|u\| \le \lambda$, $\Theta(u)$ has the following form adopted from \cite{abu2005nearly},
\bea\label{eq:Thetau}
\Theta(u) = 2\int_{0}^{u}\lambda \text{tanh}^{-1}\Big(\frac{z}{\lambda}\Big)\gamma_1\dt z 
\eea
where $\text{tanh}^{-1}(\cdot)$ denotes the inverse of the hyperbolic tangent function. One also needs to note that the states $s$ and control $u$ are not coupled in the safety constraints.

\bibliographystyle{IEEEtran}
\bibliography{root}

\begin{thebibliography}{10}
\providecommand{\url}[1]{#1}
\csname url@samestyle\endcsname
\providecommand{\newblock}{\relax}
\providecommand{\bibinfo}[2]{#2}
\providecommand{\BIBentrySTDinterwordspacing}{\spaceskip=0pt\relax}
\providecommand{\BIBentryALTinterwordstretchfactor}{4}
\providecommand{\BIBentryALTinterwordspacing}{\spaceskip=\fontdimen2\font plus
\BIBentryALTinterwordstretchfactor\fontdimen3\font minus
  \fontdimen4\font\relax}
\providecommand{\BIBforeignlanguage}[2]{{%
\expandafter\ifx\csname l@#1\endcsname\relax
\typeout{** WARNING: IEEEtran.bst: No hyphenation pattern has been}%
\typeout{** loaded for the language `#1'. Using the pattern for}%
\typeout{** the default language instead.}%
\else
\language=\csname l@#1\endcsname
\fi
#2}}
\providecommand{\BIBdecl}{\relax}
\BIBdecl

\bibitem{sutton2018reinforcement}
R.~S. Sutton and A.~G. Barto, \emph{Reinforcement learning: An
  introduction}.\hskip 1em plus 0.5em minus 0.4em\relax MIT Press, 2018.

\bibitem{kiumarsi2017optimal}
B.~Kiumarsi, K.~G. Vamvoudakis, H.~Modares, and F.~L. Lewis, ``Optimal and
  autonomous control using reinforcement learning: A survey,'' \emph{IEEE
  transactions on neural networks and learning systems}, vol.~29, no.~6, pp.
  2042--2062, 2017.

\bibitem{vrabie2013optimal}
D.~Vrabie, K.~G. Vamvoudakis, and F.~L. Lewis, \emph{Optimal adaptive control
  and differential games by reinforcement learning principles}.\hskip 1em plus
  0.5em minus 0.4em\relax IET, 2013, vol.~2.

\bibitem{kamalapurkar2018reinforcement}
R.~Kamalapurkar, P.~Walters, J.~Rosenfeld, and W.~Dixon, \emph{Reinforcement
  learning for optimal feedback control}.\hskip 1em plus 0.5em minus
  0.4em\relax Springer, 2018.

\bibitem{lin2009uav}
L.~Lin and M.~A. Goodrich, ``Uav intelligent path planning for wilderness
  search and rescue,'' in \emph{2009 IEEE/RSJ International Conference on
  Intelligent Robots and Systems}.\hskip 1em plus 0.5em minus 0.4em\relax IEEE,
  2009, pp. 709--714.

\bibitem{li2018planning}
B.~Li, S.~Patankar, B.~Moridian, and N.~Mahmoudian, ``Planning large-scale
  search and rescue using team of uavs and charging stations,'' in \emph{2018
  IEEE International Symposium on Safety, Security, and Rescue Robotics
  (SSRR)}.\hskip 1em plus 0.5em minus 0.4em\relax IEEE, 2018, pp. 1--8.

\bibitem{emerson1990temporal}
E.~A. Emerson, ``Temporal and modal logic,'' in \emph{Formal Models and
  Semantics}.\hskip 1em plus 0.5em minus 0.4em\relax Elsevier, 1990, pp.
  995--1072.

\bibitem{heemels2012introduction}
W.~Heemels, K.~H. Johansson, and P.~Tabuada, ``An introduction to
  event-triggered and self-triggered control,'' in \emph{2012 IEEE 51st IEEE
  Conference on Decision and Control (CDC)}.\hskip 1em plus 0.5em minus
  0.4em\relax IEEE, 2012, pp. 3270--3285.

\bibitem{ames2016control}
A.~D. Ames, X.~Xu, J.~W. Grizzle, and P.~Tabuada, ``Control barrier function
  based quadratic programs for safety critical systems,'' \emph{IEEE
  Transactions on Automatic Control}, vol.~62, no.~8, pp. 3861--3876, 2016.

\bibitem{konda2019provably}
R.~Konda, E.~Squires, P.~Pierpaoli, M.~Egerstedt, and S.~Coogan,
  ``Provably-safe autonomous navigation of traffic circles,'' in \emph{2019
  IEEE Conference on Control Technology and Applications (CCTA)}.\hskip 1em
  plus 0.5em minus 0.4em\relax IEEE, 2019, pp. 876--881.

\bibitem{kimmel2015active}
M.~Kimmel and S.~Hirche, ``Active safety control for dynamic human-robot
  interaction,'' in \emph{2015 IEEE/RSJ International Conference on Intelligent
  Robots and Systems (IROS)}.\hskip 1em plus 0.5em minus 0.4em\relax IEEE,
  2015, pp. 4685--4691.

\bibitem{fisac2018general}
J.~F. Fisac, A.~K. Akametalu, M.~N. Zeilinger, S.~Kaynama, J.~Gillula, and
  C.~J. Tomlin, ``A general safety framework for learning-based control in
  uncertain robotic systems,'' \emph{IEEE Transactions on Automatic Control},
  vol.~64, no.~7, pp. 2737--2752, 2018.

\bibitem{greene2020sparse}
M.~L. Greene, P.~Deptula, S.~Nivison, and W.~E. Dixon, ``Sparse learning-based
  approximate dynamic programming with barrier constraints,'' \emph{IEEE
  Control Systems Letters}, vol.~4, no.~3, pp. 743--748, 2020.

\bibitem{sun2020continuous}
C.~Sun and K.~G. Vamvoudakis, ``Continuous-time safe learning with temporal
  logic constraints in adversarial environments,'' in \emph{2020 American
  Control Conference (ACC)}.\hskip 1em plus 0.5em minus 0.4em\relax IEEE, 2020,
  pp. 4786--4791.

\bibitem{muniraj2018enforcing}
D.~Muniraj, K.~G. Vamvoudakis, and M.~Farhood, ``Enforcing signal temporal
  logic specifications in multi-agent adversarial environments: A deep
  q-learning approach,'' in \emph{2018 IEEE Conference on Decision and Control
  (CDC)}.\hskip 1em plus 0.5em minus 0.4em\relax IEEE, 2018, pp. 4141--4146.

\bibitem{saha2014automated}
I.~Saha, R.~Ramaithitima, V.~Kumar, G.~J. Pappas, and S.~A. Seshia, ``Automated
  composition of motion primitives for multi-robot systems from safe ltl
  specifications,'' in \emph{2014 IEEE/RSJ International Conference on
  Intelligent Robots and Systems}.\hskip 1em plus 0.5em minus 0.4em\relax IEEE,
  2014, pp. 1525--1532.

\bibitem{sadigh2016safe}
D.~Sadigh and A.~Kapoor, ``Safe control under uncertainty with probabilistic
  signal temporal logic,'' in \emph{Proc. of Robotics: Science and Systems},
  2016.

\bibitem{heemels2012periodic}
W.~H. Heemels, M.~Donkers, and A.~R. Teel, ``Periodic event-triggered control
  for linear systems,'' \emph{IEEE Transactions on Automatic Control}, vol.~58,
  no.~4, pp. 847--861, 2012.

\bibitem{girard2014dynamic}
A.~Girard, ``Dynamic triggering mechanisms for event-triggered control,''
  \emph{IEEE Transactions on Automatic Control}, vol.~60, no.~7, pp.
  1992--1997, 2014.

\bibitem{cheng2017event}
T.-H. Cheng, Z.~Kan, J.~R. Klotz, J.~M. Shea, and W.~E. Dixon,
  ``Event-triggered control of multiagent systems for fixed and time-varying
  network topologies,'' \emph{IEEE Transactions on Automatic Control}, vol.~62,
  no.~10, pp. 5365--5371, 2017.

\bibitem{donkers2010output}
M.~Donkers and W.~Heemels, ``Output-based event-triggered control with
  guaranteed $\mathcal{L}_2$-gain and improved event-triggering,'' in
  \emph{49th IEEE Conference on Decision and Control (CDC)}.\hskip 1em plus
  0.5em minus 0.4em\relax IEEE, 2010, pp. 3246--3251.

\bibitem{vamvoudakis2014event}
K.~G. Vamvoudakis, ``Event-triggered optimal adaptive control algorithm for
  continuous-time nonlinear systems,'' \emph{IEEE/CAA Journal of Automatica
  Sinica}, vol.~1, no.~3, pp. 282--293, 2014.

\bibitem{vamvoudakis2016event}
K.~G. Vamvoudakis and H.~Ferraz, ``Event-triggered h-infinity control for
  unknown continuous-time linear systems using q-learning,'' in \emph{2016 IEEE
  55th Conference on Decision and Control (CDC)}.\hskip 1em plus 0.5em minus
  0.4em\relax IEEE, 2016, pp. 1376--1381.

\bibitem{xu2019controller}
Z.~Xu, F.~M. Zegers, B.~Wu, W.~Dixon, and U.~Topcu, ``Controller synthesis for
  multi-agent systems with intermittent communication. a metric temporal logic
  approach,'' in \emph{2019 57th Annual Allerton Conference on Communication,
  Control, and Computing (Allerton)}.\hskip 1em plus 0.5em minus 0.4em\relax
  IEEE, 2019, pp. 1015--1022.

\bibitem{kontoudis2020online}
G.~P. Kontoudis, Z.~Xu, and K.~G. Vamvoudakis, ``Online, model-free motion
  planning in dynamic environments: An intermittent, finite horizon approach
  with continuous-time q-learning,'' in \emph{2020 American Control Conference
  (ACC)}.\hskip 1em plus 0.5em minus 0.4em\relax IEEE, 2020, pp. 3873--3878.

\bibitem{Lahijanian2016}
M.~{Lahijanian}, M.~R. {Maly}, D.~{Fried}, L.~E. {Kavraki}, H.~{Kress-Gazit},
  and M.~Y. {Vardi}, ``Iterative temporal planning in uncertain environments
  with partial satisfaction guarantees,'' \emph{IEEE Transactions on Robotics},
  vol.~32, no.~3, pp. 583--599, 2016.

\bibitem{Xu2019}
Z.~{Xu}, M.~{Ornik}, A.~A. {Julius}, and U.~{Topcu}, ``Information-guided
  temporal logic inference with prior knowledge,'' in \emph{2019 American
  Control Conference (ACC)}, 2019, pp. 1891--1897.

\bibitem{abu2005nearly}
M.~Abu-Khalaf and F.~L. Lewis, ``Nearly optimal control laws for nonlinear
  systems with saturating actuators using a neural network hjb approach,''
  \emph{Automatica}, vol.~41, no.~5, pp. 779--791, 2005.

\bibitem{borgers2014event}
D.~P. Borgers and W.~M.~H. Heemels, ``Event-separation properties of
  event-triggered control systems,'' \emph{IEEE Transactions on Automatic
  Control}, vol.~59, no.~10, pp. 2644--2656, 2014.

\bibitem{mazo2011decentralized}
M.~Mazo and P.~Tabuada, ``Decentralized event-triggered control over wireless
  sensor/actuator networks,'' \emph{IEEE Transactions on Automatic Control},
  vol.~56, no.~10, pp. 2456--2461, 2011.

\bibitem{tabuada2007event}
P.~Tabuada, ``Event-triggered real-time scheduling of stabilizing control
  tasks,'' \emph{IEEE Transactions on Automatic Control}, vol.~52, no.~9, pp.
  1680--1685, 2007.

\bibitem{modares2014optimal}
H.~Modares and F.~L. Lewis, ``Optimal tracking control of nonlinear
  partially-unknown constrained-input systems using integral reinforcement
  learning,'' \emph{Automatica}, vol.~50, no.~7, pp. 1780--1792, 2014.

\bibitem{khalil2002nonlinear}
H.~K. Khalil and J.~W. Grizzle, \emph{Nonlinear systems}.\hskip 1em plus 0.5em
  minus 0.4em\relax Prentice hall Upper Saddle River, NJ, 2002, vol.~3.

\bibitem{yang2020safe}
Y.~Yang, K.~G. Vamvoudakis, H.~Modares, Y.~Yin, and D.~C. Wunsch, ``Safe
  intermittent reinforcement learning with static and dynamic event
  generators,'' \emph{IEEE Transactions on Neural Networks and Learning
  Systems}, 2020.

\bibitem{hornik1990universal}
K.~Hornik, M.~Stinchcombe, and H.~White, ``Universal approximation of an
  unknown mapping and its derivatives using multilayer feedforward networks,''
  \emph{Neural networks}, vol.~3, no.~5, pp. 551--560, 1990.

\bibitem{vamvoudakis2017event}
K.~G. Vamvoudakis, A.~Mojoodi, and H.~Ferraz, ``Event-triggered optimal
  tracking control of nonlinear systems,'' \emph{International Journal of
  Robust and Nonlinear Control}, vol.~27, no.~4, pp. 598--619, 2017.

\bibitem{haddad2006impulsive}
W.~M. Haddad, V.~Chellaboina, and S.~G. Nersesov, \emph{Impulsive and hybrid
  dynamical systems: stability, dissipativity, and control}.\hskip 1em plus
  0.5em minus 0.4em\relax Princeton University Press, 2006, vol.~49.

\bibitem{hespanha2008lyapunov}
J.~P. Hespanha, D.~Liberzon, and A.~R. Teel, ``Lyapunov conditions for
  input-to-state stability of impulsive systems,'' \emph{Automatica}, vol.~44,
  no.~11, pp. 2735--2744, 2008.

\end{thebibliography}
\end{document}